\documentclass[12pt]{article}

\usepackage{graphicx, amssymb, amsmath, amsthm, braket, tikz, color, subcaption, cite, setspace}
\usepackage[colorlinks = true]{hyperref}
\usepackage[margin = 2cm]{geometry}
\newtheorem{lemma}{Lemma}
\newtheorem{theorem}{Theorem}
\newtheorem{corollary}{Corollary}

\title{Perfect state transfer using Markovian quantum walk}
\author{
	Supriyo Dutta \\ 
	\small{Department of Mathematics,}\\ 
	\small{National Institute of Technology Agartala,} \\ 
	\small{Jirania, West Tripura, India.}\\ 
	\small{\texttt{dosupriyo@gmail.com}}}
\date{}

\begin{document}

	\maketitle 

	\begin{abstract}
		The quantum Perfect State Transfer (PST) is a fundamental tool of quantum communication in a network. It is not easy to achieve in practice. The original idea of PST depends on the fundamentals of the continuous-time quantum walk. A path graph with at most three vertices allows PST based on continuous-time quantum walk. Based on the Markovian quantum walk, we introduce a significantly powerful method for PST in this article. We establish PST between the extreme vertices of a path graph of arbitrary length. Moreover, any pair of symmetric vertices in a path graph allows PST under Markovian quantum walks. We extend our investigations for the cycle graphs. The cycle graphs with more than $4$ vertices do not allow the PST based on the continuous-time quantum walk. In contrast, a cycle graph with $2m$ vertices exhibits PST based on Markovian quantum walk between the vertices $j$ and $j + m$ for $j = 0, 1, \dots (m - 1)$, where $m > 0$ is an integer.
	\end{abstract}

	\tableofcontents
	
	\begin{onehalfspace}
		
	\newpage 
	\section{Introduction}

		The transfer of a quantum state from one location in a quantum network to another without interrupting the encoded information is a crucial task in quantum information processing \cite{benenti2019principles, wilde2013quantum, bose2003quantum, christandl2004perfect, chapman2016experimental, osborne2006statics, plenio2004dynamics, nikolopoulos2004electron, divincenzo2000physical}. Considering the advantage of coupling between neighbouring qubits, we can physically transport quantum information across a network, which leads us to the idea of quantum Perfect State Transfer (PST) \cite{yung2006quantum, christandl2005perfect}. In quantum computing, quantum walks have been a successful tool with many algorithmic applications. The problem of perfect state transfer in quantum walks is a graph-theoretic in nature. Therefore, we begin this article with some basic definitions of graph theory.
		
		A graph $G = (V(G), E(G))$ is a combination of a vertex set $V(G)$ and an edge set $E(G) \subset V(G) \times V(G)$ \cite{west2001introduction}. An edge is said to be a directed edge if there is a direction in it. A directed graph is a graph whose all edges are directed edges. An edge joining a vertex with itself is said to be a loop. A simple graph has no loops and no directed edges in it. We can convert a simple graph to a directed graph by assigning two opposite directions on every edge. If a graph $G$ contains $n$ vertices, we mark them with $0, 1, \dots (n - 1)$.  The adjacency matrix of a graph $G$ \cite{bapat2010graphs} $A(G) = (a_{i, j})_{n \times n}$ is represented by 
		\begin{equation}
			a_{i, j} = \begin{cases} 1 & ~\text{if}~ (i, j) \in E(G), \\0 &  ~\text{if}~ (i, j) \notin E(G). \end{cases} 
		\end{equation}
		We utilize two simple graphs in this article, which are the path and cycle graphs. A path graph with $n$ vertices is denoted by $P_n$ with vertices $0, 1, \dots (n - 1)$ and edges $\{(0, 1), (1, 2), \dots ((n - 2), (n - 1))\}$. Here, $0$ and $(n - 1)$ are the extreme vertices of $P_n$. A cycle graph $C_n$ is a graph $n$ vertices $0, 1, \dots (n - 1)$ and edges $\{(0, 1), (1, 2), \dots ((n - 2), (n - 1)), ((n - 1), 0)\}$.
		
		In the literature of quantum information and computation, there are different proposals of quantum walk on graphs, for instance, the continuous-time quantum walk \cite{farhi1998quantum, childs2009universal, mulken2011continuous, godsil2008periodic}, discrete-time quantum walk \cite{aharonov2001quantum, ambainis2003quantum}, Szegedy quantum walk \cite{szegedy2004quantum}, and many others \cite{portugal2013quantum, venegas2012quantum, konno2008quantum, ambainis2007quantum, du2011search, buhrman2004quantum}. The dynamics of each of these quantum walks is determined by a unitary matrix, for instance, the unitary matrix $\exp(\iota A(G) t)$ determines continuous-time quantum walk on a graph. Let $\ket{j} \in \mathcal{H}^n$ be the initial state of a quantum walker starting from vertex $j$ at time $t = 0$. After time $t$ state of the walker is $\exp(\iota A(G) t) \ket{j}$. The continuous-time quantum walk generates a PST between the vertex $j$ and $k$ at time $t_0$ \cite{kendon2011perfect, godsil2012state, godsil2012number, cheung2011perfect} if we have 
		\begin{equation}\label{continuous_time_walkbased_PST}
			\ket{k} = \exp(\iota A(G) t_0) \ket{j}.
		\end{equation} 
		The idea of PST is generalized in a similar fashion for other quantum walks \cite{kurzynski2011discrete, zhan2014perfect, chan2021pretty, shang2019quantum}. 
		
		In this article, we discuss PST related to the Markov chain-based quantum walk \cite{segawa2013localization, balu2017probability} which is a variant of the Szegedy quantum walk. The original idea of Szegedy quantum walk \cite{szegedy2004quantum, szegedy2004spectra} defines reflection operators on a bipartite graph. The Szegedy quantum walk not only quantizes the Markov chains, but it also provides a quadratic speedup in the hitting time, which is the required time for finding a marked vertex in a graph. This model of quantum walks is a fundamental tool for complex quantum algorithms, for instance, the element distinctness, triangle finding, etc.  In addition, recent research \cite{portugal2017staggered} suggests that the physical implementation of these walks may be conducted by time-independent Hamiltonians. Hence, the Szegedy quantum walks can be more easily mapped to certain physical systems, which include the superconducting resonators or atomic lattices, compared to the other models. The Markovian quantum walk is a variant of Szegedy quantum walk, which does not need a bipartite graph. Therefore, we coin it with a different name.
		
		Different graphs satisfy PST related to different quantum walks. Not all simple graphs support PST. For instance, the path graphs $P_2$ and $P_3$ allow PST based on continuous-time quantum walk between its extreme vertices. When the number of vertices $n > 3$, the path graph $P_n$ does not allow PST based on continuous-time quantum walk \cite{christandl2004perfect}. The idea of quantum walks on cycle graphs is a well-studied topic in literature \cite{adamczak2007non, travaglione2002implementing, tregenna2003controlling}. Only the cycle graph with four vertices allows PST based on the continuous-time quantum walk \cite{christandl2004perfect, kendon2011perfect}. There are many other graphs with more edges supporting PST in a greater distance, such as hypercube graphs, cubelike graphs \cite{cheung2011perfect}, Cayley graphs \cite{pal2017perfect, tan2019perfect}, etc. 
		The path graph $P_n$ has $(n - 1)$ edges which is the minimum number of edges to connect two extreme vertices at the maximum distance. Also, the cycle graph $C_n$ has $n$ edges. Therefore, the distance between the extreme vertices in $P_n$ is $n$. The maximum distance between two extreme vertices of $C_{2n}$ is $n$. Technically, the vertices are spin objects and edges are the interactions between them. It is difficult to control a large number of interactions between spin objects when we want to communicate between two vertices at a distance. This is why we consider the path and cycle graphs as the ideal candidates for performing PST.
		
		To the best of our knowledge, this is the first article in literature reporting PST based on the Markovian quantum walk on the cycle and path graphs in arbitrary distance. Whereas the other quantum walks produce PST on these graphs at a limited distance. Our observations are also quite interesting. The new PST allows us to communicate between the end vertices of a path graph of arbitrary length. Moreover, it exhibits PST between the pairs of symmetric vertices. Therefore, in cycle and path graphs PST based on the Markovian quantum walk is more efficient than the original idea of PST based on a continuous-time quantum walk in quantum communication. We extended our investigation for the cycle graph also. We prove that any cycle graph with an even number of vertices permits PST based on the Markovian quantum walk.
		
		We distribute this work as follows. In Section 2, we introduce Markovian quantum walk and PST using it. Here, we develop a basis of $\mathcal{H}^n$ different from the computational basis to perform the quantum walk. We bind up this section with three lemmas which are essential to establish PSTs. Sections 3 and 4 are dedicated to state transfer on the cycle, and the path graphs, respectively. In these sections, we work out the state after $t$ time steps of evolutions starting from an arbitrary initial vertex $j$. It assists us in identifying a pair of vertices allowing PST and periodic vertices. Then we conclude this article.

	\section{Markovian quantum walk and perfect quantum state transfer}

		The random walk \cite{lawler2010random} on a graph is a fundamental mathematical tool for modelling and simulating complex problems and natural phenomena. Szegedy developed a general method for quantizing a random walk to create a discrete-time quantum walk. To develop the idea of a quantum walk, we first generate a directed graph from a given simple graph $G$. On every edge, we assign two opposite orientations. For example, $(j, k)$ and $(k, j)$ denote two oppositely oriented edges from the vertex $j$ to $k$, and from $k$ to $j$, respectively. The out-degree of a vertex $j$ is denoted by $d_j$ which is the number of outgoing edges from $j$. We define 
		\begin{equation}\label{probability_distn}
			p_{j, k} = 
			\begin{cases} 
				\frac{1}{d_j} & ~\text{if}~ (j, k) \in E(G),\\
				0 &  ~\text{if}~ (j, k) \notin E(G);
			\end{cases}
		\end{equation}
		such that $\sum_{k \in V(G)} p_{j, k} = 1$. As $G$ is connected $d_j \neq 0$ for any $v \in V$ and $0 \leq p_{j, k} \leq 1$. The probability transition matrix $P = (p_{jk})_{n \times n}$ leads to a random walk on $G$.
		
		For every vertex $j$, we assign an element of the computational basis of $n$ dimensional vector space $\mathcal{H}^n$, say $\ket{j}$, where $\ket{j} = [0, 0, \dots 1 (j\text{-th position}) \dots 0]^\dagger$. Corresponding to a directed edge $(j, k)$ from $j$ to $k$, we assign a vector $\ket{jk} = \ket{j} \otimes \ket{k}$, where $\ket{jk} \in \mathcal{H}^n \otimes \mathcal{H}^n$. A superposition of the vectors representing the edges outgoing from the vertex $j$ is 
		\begin{equation}\label{psi_j}
			\ket{\psi_j} = \sum_{k \in V} \sqrt{p_{j, k}} \ket{jk} = \ket{j} \otimes \left(\sum_{k \in V} 	\sqrt{p_{jk}} \ket{k} \right).
		\end{equation}
		
		We can establish that the vectors in $\Psi = \{\ket{\psi_j}: j \in V(G)\}$ are orthonormal because $\braket{\psi_j | \psi_j} = 1$ and $\braket{\psi_j | \psi_k} = 0$ for $j \neq k$. Therefore, the vectors in $\Psi$ are linearly independent. The number of vectors in $\Psi$ is $n$, which is the dimension of $\mathcal{H}^n$. Thus, the vector subspace $\mathcal{L}(\Psi)$ spanned by $\Psi$ is isomorphic to $\mathcal{H}^n$. It leads us to define a linear transformation $T: \mathcal{L}(\Psi) \rightarrow \mathcal{H}^n$ such that $T \ket{\psi_j} = \ket{j}$ for $j \in V(G)$. It can be represented by a matrix $T$, such that, $T = \sum_{j \in V(G)} \ket{j} \bra{\psi_j}$. Also, $T^\dagger = \sum_{j \in V(G)} \ket{\psi_j} \bra{j}$, such that $T^\dagger \ket{j} = \ket{\psi_j}$. Therefore, any walk on $\mathcal{L}(\Psi)$ is equivalent to a walk on $\mathcal{H}^n$. It allows us to restrict our discussion on the quantum walk and state transfer on the space $\mathcal{L}(\Psi)$ only. 
		
		Corresponding to vertex $j \in V(G)$ we now assign a basis state $\ket{\psi_j} \in \Psi$. The unitary matrix leading the Markovian quantum walk \cite{segawa2013localization, balu2017probability} acting on $\mathcal{L}(\Psi)$ is defined by 
		\begin{equation}\label{unitary_matrix}
			U = S (2 \Pi - I).
		\end{equation} 
		Here, $\Pi = \sum_{j \in V} \ket{\psi_j} \bra{\psi_j}$ denotes the linear sum of all orthogonal projectors on $\Psi$. Also, $S = \sum_{j \in V} \sum_{k \in V} \ket{jk} \bra{kj}$ is a SWAP operator, and $I$ is the identity matrix of order $n^2$.
		
		Now, we are in a position to define PST based on the Markovian quantum walk. Suppose the initial state of the walker is $\ket{\psi_j}$, which starts at vertex $j$ at time $t = 0$. Using Born rule, we state that the probability of getting the walker at vertex $k$ is $P_t(k) = |\braket{\psi_k | U^t | \psi_j}|^2$, after time $t$. We also say that there is a PST between vertices $j$ and $k$ at time $t$ if 
		\begin{equation}
			\ket{\psi_k} = U^t \ket{\psi_j},
		\end{equation}
		that is,  $P_t(k) = 1$. We define a vertex $j \in V$ is a periodic vertex at time $t$ if $\ket{\psi_j} = U^t \ket{\psi_j}$. 
		
		Below we mention three lemmas which illuminate the characteristics of the evolution operator $U$. In the next section, we utilize them in our calculations.
		
		\begin{lemma} \label{App_Lemma_1} 
			Let $G$ be a path or a cycle graph with vertices $0, 1, 2, \dots (n - 1)$. For $j \neq 0$ and $j \neq (n - 1)$ define $\ket{\psi_j} = \frac{1}{\sqrt{2}} \left( \ket{j(j + 1)} + \ket{j(j - 1)} \right)$ and $\Pi = \sum_{j \in V} \ket{\psi_j} \bra{\psi_j}$. Then
			$$S(2\Pi - I) \ket{\psi_j} = \frac{1}{\sqrt{2}} \left(\ket{(j + 1)j} + \ket{(j - 1)j} \right).$$ 
		\end{lemma}
		\begin{proof}
			\begin{equation}
				\begin{split}
					& (2\Pi - I) \ket{\psi_j} = 2\Pi \ket{\psi_j} - \ket{\psi_j} = 2\ket{\psi_j} \braket{\psi_j | \psi_j} - \ket{\psi_j} = 2\ket{\psi_j} - \ket{\psi_j} = \ket{\psi_j} \\
					\text{or}~ & S (2\Pi - I) \ket{\psi_j} = S \ket{\psi_j} = \frac{1}{\sqrt{2}} S (\ket{j(j+1)} + \ket{j(j - 1)}) = \frac{1}{\sqrt{2}}(\ket{(j+1)j} + \ket{(j - 1)j}).
				\end{split}
			\end{equation}
		\end{proof}
		
		\begin{lemma}\label{App_Lemma_2} 
			Let $G$ be a path or a cycle graph with vertices $0, 1, 2, \dots (n - 1)$. For $j \neq (n - 1)$ we have 
			$$S(2\Pi - I) \ket{(j + 1) j} = \ket{(j + 2)(j + 1)}.$$
			Here, $\ket{\psi_j} = \frac{1}{\sqrt{2}} \left( \ket{j(j + 1)} + \ket{j(j - 1)} \right)$ and $\Pi = \sum_{j \in V} \ket{\psi_j} \bra{\psi_j}$.
		\end{lemma}
		\begin{proof}
			\begin{equation}
				\begin{split}
					(2\Pi - I) \ket{(j + 1) j} & = 2\Pi \ket{(j + 1) j} - I \ket{(j + 1) j} \\
					& = 2 \ket{\psi_{(j + 1)}} \braket{\psi_{(j + 1)}| (j + 1) j} - \ket{(j + 1) j} = 2 \frac{1}{\sqrt{2}} \ket{\psi_{(j + 1)}}- \ket{(j + 1) j},
				\end{split}
			\end{equation}
			since $\braket{\psi_{(j + 1)}| (j + 1) j} = \frac{1}{\sqrt{2}}$. Applying equation (\ref{psi_j}) we have
			\begin{equation}
				\begin{split}
					& S(2\Pi - I) \ket{(j + 1) j} = S \left( 2 \frac{1}{\sqrt{2}} \frac{1}{\sqrt{2}}(\ket{(j+1)(j+2)} + \ket{(j + 1)j}) - \ket{(j + 1) j} \right) \\
					& = S \left( \ket{(j+1)(j+2)} + \ket{(j + 1)j} - \ket{(j + 1) j} \right) = S\ket{(j+1)(j+2)} = \ket{(j + 2)(j + 1)}.
				\end{split}
			\end{equation}
		\end{proof}
		
		\begin{lemma}\label{App_Lemma_3} 
			Let $G$ be a path or a cycle graph with vertices $0, 1, 2, \dots (n - 1)$. For $j \neq 0$ we have 
			$$S(2\Pi - I) \ket{(j - 1) j} = \ket{(j - 2)(j - 1)}.$$
			Here, $\ket{\psi_j} = \frac{1}{\sqrt{2}} \left( \ket{j(j + 1)} + \ket{j(j - 1)} \right)$ and $\Pi = \sum_{j \in V} \ket{\psi_j} \bra{\psi_j}$.
		\end{lemma}
		\begin{proof}
			\begin{equation}
				\begin{split}
					(2\Pi - I) \ket{(j - 1) j} & = 2\Pi \ket{(j - 1) j} - \ket{(j - 1) j} \\
					& = 2 \ket{\psi_{(j - 1)}} \braket{\psi_{(j - 1)}| (j - 1) j} - \ket{(j - 1) j} = 2 \frac{1}{\sqrt{2}} \ket{\psi_{(j - 1)}}- \ket{(j - 1) j},
				\end{split}
			\end{equation}
			since $\braket{\psi_{(j - 1)}| (j - 1) j} = \frac{1}{\sqrt{2}}$. Applying equation (\ref{psi_j}) we have
			\begin{equation}
				\begin{split}
					& S(2\Pi - I) \ket{(j - 1) j} = S \left( 2 \frac{1}{\sqrt{2}} \frac{1}{\sqrt{2}}(\ket{(j - 1)j} + \ket{(j - 1)(j - 2)}) - 	\ket{(j - 1) j} \right) \\
					& = S \left( \ket{(j-1)(j-2)} + \ket{(j - 1)j} - \ket{(j - 1) j} \right) = S \ket{(j-1)(j-2)} = \ket{(j - 2)(j - 1)}.
				\end{split}
			\end{equation}
		\end{proof}

		\section{State transfer on cycle graph}
		\label{PST_cycle}
		
		A cycle graph $C_n$ with $n$ vertices is given by a vertex set $V(C_n) = \{0, 1, 2, \dots (n - 1)\}$ and an edges set $E(C_n) = \{((j - 1), j): j = 1, 2, \dots (n - 1)\} \cup \{((n - 1), 0)\}$. In other words, two vertices $j$ and $k$ are connected in $C_n$ if and only if $j - k \equiv 1 (\mod n)$. Therefore, for every vertex $j$ we can define 
		\begin{equation}
			p_{j, k} = \begin{cases} \frac{1}{2}  & ~\text{if}~ j - k \equiv 1 (\mod n), \\ 0 & ~\text{otherwise.} 	\end{cases} 
		\end{equation}
		For example, consider a cycle graph with $6$ vertices $C_6$ depicted in Figure \ref{undirected_6_cycle}. Considering two opposite orientations on $C_6$ we get a directed graph depicted in Figure \ref{cycle_6}. 
		\begin{figure}
			\begin{subfigure}[t]{.48\textwidth}
				\centering 
				\begin{tikzpicture}
					\draw [fill] (0, 0) circle [radius=0.05 cm];
					\node [below] at (0, 0) {$0$};
					\draw [fill] (1, 1) circle [radius=0.1 cm];
					\node [below right] at (1, 1) {$1$};
					\draw [fill] (-1, 1) circle [radius=0.1 cm];
					\node [below left] at (-1, 1) {$5$};
					\draw [fill] (1, 2) circle [radius=0.1 cm];
					\node [above right] at (1, 2) {$2$};
					\draw [fill] (-1, 2) circle [radius=0.1 cm];
					\node [above left] at (-1, 2) {$4$};
					\draw [fill] (0, 3) circle [radius=0.1 cm];
					\node [above] at (0, 3) {$3$};
					\draw (0, 0) -- (1, 1) -- (1, 2) -- (0, 3) -- (-1, 2) -- (-1, 1) -- (0, 0);
				\end{tikzpicture}
				\caption{}
				\label{undirected_6_cycle}
			\end{subfigure}
			\hspace{.5cm} 	
			\begin{subfigure}[t]{.48\textwidth}
				\centering 
				\begin{tikzpicture}
					\draw [fill] (0, 0) circle 	[radius=0.05 cm];
					\node [below] at (0, 0) 	{$0$};
					\draw [fill] (1, 1) circle 	[radius=0.05 cm];
					\node [below right] at (1, 1) {$1$};
					\draw [fill] (-1, 1) 	circle [radius=0.05 cm];
					\node [below left] at (-1, 1) {$5$};
					\draw [fill] (1, 2) circle [radius=0.05 cm];
					\node [above right] at (1, 2) {$2$};
					\draw [fill] (-1, 2) 	circle [radius=0.05 cm];
					\node [above left] at (-1, 2) {$4$};
					\draw [fill] (0, 3) circle [radius=0.05 cm];
					\node [above] at (0, 3) {$3$};
					\draw [out=0,in= -90, ->]  (0, 0) to  (1, .9);
					\draw [out = 210, in = 60, ->] (1, 1) to (0, .1);
					\draw [out = 45, in = 315, ->] (1, 1) to (1.05, 1.95);
					\draw [out = 225, in = 	135, ->] (1, 2) to (1, 1.1);
					\draw [out = 90, in = 0, ->] (1, 2) to (.1, 3);
					\draw [out = 300, in = 	150, ->] (0, 3) to (.9, 2);
					\draw [out = 180, in = 90, ->] (0, 3) to (-1, 2.1);
					\draw [out = 30, in = 240, ->] (-1, 2) to (0, 2.9);
					\draw [out = 225, in = 	135, ->] (-1, 2) to (-1, 1.1);
					\draw [out = 45, in = 315, ->] (-1, 1) to (-.95, 1.95);
					\draw [out = 270, in = 	180, ->] (-1, 1) to (-.1, 0);
					\draw [out = 120, in = 	330, ->] (0, 0) to (-.95, .95);
				\end{tikzpicture}
				\caption{}
				\label{cycle_6}
			\end{subfigure}
			\caption{Sub-figure (\ref{undirected_6_cycle}) presents a cycle graph $C_6$ with $6$ vertices. Note that, there is an edge between the vertices $j$ and $k$ if and only if $j - k \equiv 1(\mod 6)$. Figure (\ref{cycle_6}) is $C_6$ after orientation on the edges. Every edge in sub-figure (\ref{undirected_6_cycle}) generates two oppositely oriented edges in sub-figure (\ref{cycle_6}). Thus, every vertex is incident to two incoming edges and two outgoing edges, which means the out-degree of every vertex is $2$. Hence, $p_{j, k} = \frac{1}{2}$ for all possible edges $(j, k)$ in the cycle graph.}
		\end{figure}
		We can specify equation (\ref{psi_j}) for the cycle graphs to write the state corresponding to vertex $j$, which is as follows:
		\begin{equation}
			\ket{\psi_j} = \begin{cases} \frac{1}{\sqrt{2}}(\ket{01} + \ket{0(n - 1)}) & ~\text{if}~ j = 0; \\ \frac{1}{\sqrt{2}} \ket{j(j + 1)} + \ket{j(j - 1)} & ~\text{if}~ j= 1, 2, \dots (n - 2); \\ \frac{1}{\sqrt{2}}(\ket{(n - 1)0} + \ket{(n - 1)(n - 2)}) & ~\text{if}~ j = (n - 1); \end{cases}
		\end{equation}
		for any cycle $C_n$. The above equation can be written more precisely 
		\begin{equation}\label{psi_j_cycle}
			\ket{\psi_j} = \frac{1}{\sqrt{2}} (\ket{j(j \oplus 1)} + \ket{j(j \ominus 1)}),
		\end{equation}
		where $\oplus$ and $\ominus$ denote addition modulo $n$ and subtraction modulo $n$, respectively, throughout the text. Now, we demonstrate the following characteristic of the operator $U = S(2\Pi - I)$ for $C_n$, which will play a crucial role in PST.
		
		\begin{theorem}\label{theorem_propagation_on_cycle}
			In a cycle graph $C_n$ applying $U = S(2\Pi - I)$ on state $\ket{\psi_j}$ for $t$ times we have
			$$[S(2 \Pi - I)]^t \ket{\psi_j} = \frac{1}{\sqrt{2}} \left(\ket{(j \oplus t)(j \oplus (t - 1))} + \ket{(j \ominus t)(j \ominus (t - 1))}\right),$$
			where $j = 0, 1, 2, \dots (n - 1)$.
		\end{theorem} 
		
		\begin{proof} 
			We combine the following facts:
			\begin{enumerate}
				\item 
				Lemma \ref{App_Lemma_1} suggests that for any vertex $j$ in $C_n$ we have $S(2\Pi - I) \ket{\psi_j} = \frac{1}{\sqrt{2}} \left(\ket{(j + 1)j} + \ket{(j - 1)j} \right)$, when $j \neq 0$ and $(n - 1)$. For $j = 0$, 
				\begin{equation}
					\begin{split}
						S(2\Pi - I) \ket{\psi_0} & = S(2\Pi\ket{\psi_0} - \ket{\psi_0}) = S\ket{\psi_0} = \frac{1}{\sqrt{2}} \left(\ket{10} + \ket{(n - 1)0} \right).
					\end{split}
				\end{equation}
				Similarly, for $j = n - 1$ we can verify that 
				\begin{equation}
					S(2\Pi - I) \ket{\psi_{n - 1}} = \frac{1}{\sqrt{2}} \left(\ket{0(n - 1)} + \ket{(n - 2)(n - 1)} \right).
				\end{equation}
				Combining all these, we have
				\begin{equation}\label{cycle_lemma_1}
					S(2\Pi - I) \ket{\psi_j} = \frac{1}{\sqrt{2}} \left(\ket{(j \oplus 1)j} + \ket{(j \ominus 1)j} \right).
				\end{equation}
				\item 
				From Lemma \ref{App_Lemma_2} we state that for any vertex $j \neq 0$ and $j \neq (n -1)$ in $C_n$ we have $S(2\Pi - I) \ket{(j + 1) j} = \ket{(j + 2)(j + 1)}$. Also, for $j = n - 1$ we have 
				\begin{equation}
					\begin{split}
						S(2\Pi - I) \ket{0(n - 1)} & = S \left(2\Pi \ket{0(n - 1)} - \ket{0(n - 1)} \right) \\
						& = S\left(2\frac{1}{\sqrt{2}} \ket{\psi_0} - \ket{0(n - 1)} \right) = S\ket{01} = \ket{10}.
					\end{split}
				\end{equation}
				For $j = 0$ we have $S(2\Pi - I) \ket{1 0} = \ket{21}$. Combining we get 
				\begin{equation}\label{cycle_lemma_2}
					S(2\Pi - I) \ket{(j \oplus 1) j} = \ket{(j \oplus 2)(j \oplus 1)}.
				\end{equation}
				\item 
				Using Lemma \ref{App_Lemma_3} we can justify that $S(2\Pi - I) \ket{(j - 1) j} = \ket{(j - 2)(j - 1)}$, when $j \neq 0$ and $j \neq (n - 1)$. Considering the cases for $j = 0$ and $j = (n - 1)$ we can prove in general 
				\begin{equation}\label{cycle_lemma_3}
					S(2\Pi - I) \ket{(j \ominus 1) j} = \ket{(j \ominus 2)(j \ominus 1)}.
				\end{equation}
			\end{enumerate}
			
			Now, we are in a position to apply $S(2\Pi-I)$ multiple times on $\ket{\psi_j}$ with help of equations (\ref{cycle_lemma_1}), (\ref{cycle_lemma_2}) and (\ref{cycle_lemma_3}). Note that, 
			\begin{equation}
				\begin{split}
					[S(2 \Pi - I)]^2 \ket{\psi_j} & =  \frac{1}{\sqrt{2}} S(2\Pi - I) \left(\ket{(j \oplus 1)j} + \ket{(j \ominus 1)j} \right) \\
					& = \frac{1}{\sqrt{2}} \left( \ket{(j \oplus 2)(j \oplus 1)} + \ket{(j \ominus 2)(j \ominus 1)} \right).
				\end{split}
			\end{equation}
			Similarly,
			\begin{equation}
				\begin{split}
					[S(2 \Pi - I)]^3 \ket{\psi_j} & =  \frac{1}{\sqrt{2}} S(2\Pi - I) \left( \ket{(j \oplus 2)(j \oplus 1)} + \ket{(j \ominus 2)(j \ominus 1)} \right) \\
					& = \frac{1}{\sqrt{2}} \left( \ket{(j \oplus 3)(j \oplus 2)} + \ket{(j \ominus 3)(j \ominus 2)} \right).
				\end{split}
			\end{equation}
			Applying the principle of mathematical induction we can conclude the result.
		\end{proof} 
		
		Let the walker begins its journey from vertex $j$ on a cycle graph. Now, using Theorem \ref{theorem_propagation_on_cycle}, the probability of finding the walker at vertex $k$, at time $t$ is
		\begin{equation}
			\begin{split}
				P_t(k) & = |\braket{\psi_k | U^t | \psi_j}|^2 = |\braket{\psi_k | [S(2 \Pi - I)]^t | \psi_j}|^2\\
				& = \frac{1}{2} |\braket{\psi_k | (j \oplus t)(j \oplus (t - 1)} + \braket{\psi_k |(j \ominus t)(j \ominus (t - 1))}|^2.
			\end{split}
		\end{equation}
		Now, applying equation (\ref{psi_j_cycle}) we find the expression of $\ket{\psi_k}$. Then,
		\begin{equation}
			\begin{split}
				P_t(k) = & \frac{1}{4} |\braket{(k(k \oplus 1) + k(k \ominus 1)) | (j \oplus t)(j \oplus (t - 1)} \\
				& + \braket{(k(k \oplus 1) + k(k \ominus 1)) |(j \ominus t)(j \ominus (t - 1))}|^2 \\
				= & \frac{1}{4} |\braket{(k(k \oplus 1)| (j \oplus t)(j \oplus (t - 1)} + \braket{k(k \ominus 1)) | (j \oplus t)(j \oplus (t - 1)} \\
				& + \braket{(k(k \oplus 1) |(j \ominus t)(j \ominus (t - 1))}  + \braket{k(k \ominus 1)) |(j \ominus t)(j \ominus (t - 1))}|^2.
			\end{split}
		\end{equation}Numerical values of $P_t(k)$ are $0, 0.5$ and $1$ depending on $j, k$ and $t$. For example, on a cycle graph $C_6$, when $j = 0$ and $t = 2$ we have 
		\begin{equation}
			P_t(k) = \begin{cases} 0 & ~\text{if}~ k = 0, 1, 3, 5\\
				0.5 & ~\text{if}~ k = 2, 4.
			\end{cases} 
		\end{equation}
		We depict these probability values as a bar diagram in sub-figure \ref{cycle_probablity_6}c. 
		
		Theorem \ref{theorem_propagation_on_cycle} leads us to the idea of PST on a cycle graph with an even number of vertices. Consider the following corollaries. 		
		
		\begin{corollary}
			In a cycle graph $C_{2m}$ there is a PST between any two vertices $j$ and $j + m$ at time $t = m$ for $j = 0, 1, \dots (m - 1)$.
		\end{corollary}
		
		\begin{proof} 
			Applying Theorem \ref{theorem_propagation_on_cycle}, we have
			\begin{equation}
				[S(2\Pi - I)]^m \ket{\psi_j} =  \frac{1}{\sqrt{2}} \left[\Ket{\left(j \oplus m \right)\left(j \oplus \left(m - 1\right) \right)} + \Ket{\left(j \ominus m \right)\left(j \ominus \left(m - 1\right) \right)} \right].
			\end{equation}
			Since $0 \leq j \leq 2m - 1$, we have $j \oplus m \equiv j \ominus m$ and $j \oplus m \oplus 1 \equiv j \ominus (m - 1)$. Now, applying equation (\ref{psi_j_cycle}) we find that $[S(2\Pi - I)]^m \ket{\psi_j} = \ket{\psi_{j \oplus m}}$ Therefore, if the cycle graph has $n = 2m$ vertices then there is PST between $j$ and $j \oplus m$ at time $t = m$. 
		\end{proof} 
		
		Applying the superposition principle of quantum states, we conclude that, given any two distinct vertices $j$ and $k$ in $C_{2m}$, there are PSTs between $j$ and $j \oplus n$ as well as $k $ and $k \oplus n$ simultaneously at time $t = m$. 
		
		For example, we consider a cycle graph with $6$ vertices. At $t = 3$ there is simultaneous PST between every pair of vertices  $(0, 3), (1, 4)$, and $(2, 5)$. In Figure \ref{cycle_probablity_6}, we plot $P_t(j, k)$ with bar diagrams for $0 \leq t \leq 3$. The walker initiates walking at vertex $j = 1$ at $t = 0$ and reaches to $j = 4$ and $t = 3$ following the model of Markovian quantum walks. In contrast, the model of continuous-time quantum walks, described in equation (\ref{continuous_time_walkbased_PST}), keeps the walker at the initial vertex with full probability at different time instances. In the subfigures of Figure \ref{cycle_probablity_6}, the red and green colored bars represent the probability of getting the walker at different vertices under the continuous-time and Markovian quantum walks, respectively.
		\begin{figure} 
			\begin{subfigure}[t]{.48\textwidth}
				\centering
				\includegraphics[height = 4cm, width = 6cm]{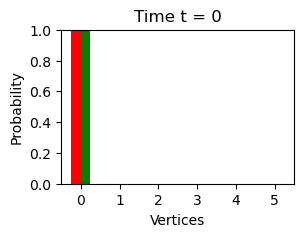}
				\caption{Probability of getting the walker at different vertices at $t = 0$.}
			\end{subfigure}
			\hspace{.5cm}
			\begin{subfigure}[t]{.48\textwidth}
				\centering
				\includegraphics[height = 4cm, width = 6cm]{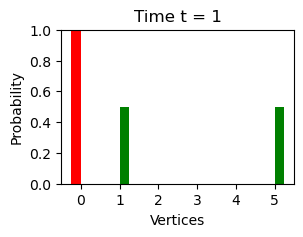}
				\caption{Probability of getting the walker at different vertices at $t = 1$.}
			\end{subfigure}\\
			\begin{subfigure}[t]{.48\textwidth}
				\centering
				\includegraphics[height = 4cm, width = 6cm]{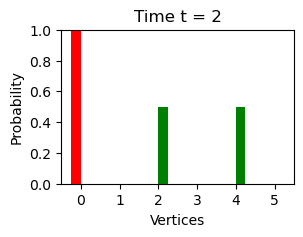}
				\caption{Probability of getting the walker at different vertices at $t = 2$.}
			\end{subfigure}
			\hspace{.5cm}
			\begin{subfigure}[t]{.48\textwidth}
				\centering
				\includegraphics[height = 4cm, width = 6cm]{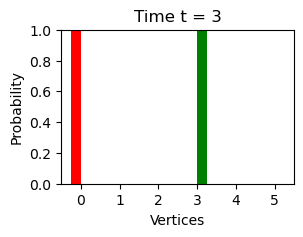}
				\caption{Probability of getting the walker at different vertices at $t = 0$.}
			\end{subfigure}
			\caption{The red and green bars in the sub-figures indicate the probability of getting the walker at vertex $k, 0 \leq k \leq 5$ in $C_6$, at different time instances $t = 0, 1, 2$ and $3$ in continuous-time and Markovian quantum walks, respectively. Suppose the walker started at vertex $1$ at $t = 0$. At $t = 3$, it reaches vertex $4$ with full probability, under the Markovian quantum walks. Therefore, there is a PST between vertices $1$ and $4$ at time $t = 3$. In contrast, the continuous-quantum walks keep the walker at vertex $1$.}
			\label{cycle_probablity_6} 
		\end{figure}
		
		\begin{corollary}
			All the vertices in $C_n$ are periodic at time $t = n$ for any $n$.
		\end{corollary}
		
		\begin{proof}
			Applying Theorem \ref{theorem_propagation_on_cycle}, we have
			\begin{equation}
				\begin{split}
					& [S(2\Pi - I)]^n \ket{\psi_j} = \frac{1}{\sqrt{2}} \left[\ket{\left(j \oplus n \right)\left(j \oplus \left(n - 1\right) \right)} + \ket{(j \ominus n) (j \ominus (n - 1))} \right] \\
					& = \frac{1}{\sqrt{2}} \left[\ket{j(j + 1)} + \ket{j(j - 1)} \right]  = \ket{\psi_j},
				\end{split}
			\end{equation}
			since $j \oplus n = j \ominus n = j$, $j \oplus (n - 1) = j - 1$, and $j \ominus (n - 1) = j + 1$. Therefore, the vertex $j$ is a periodic vertex for all $j$ at time $t = n$.
		\end{proof}
		
		For example, we consider a cycle graph with $5$ vertices. Here, all the vertices are periodic with period $t = 5$. In Figure \ref{cycle_probablity_5}, we plot $P_t(j, k)$ with bar diagrams for $0 \leq t \leq 5$ when a walker initiates walking at vertex $j = 1$.
		
		\begin{figure} 
			\begin{subfigure}[t]{.48\textwidth}
				\centering 
				\includegraphics[height = 4cm, width = 6cm]{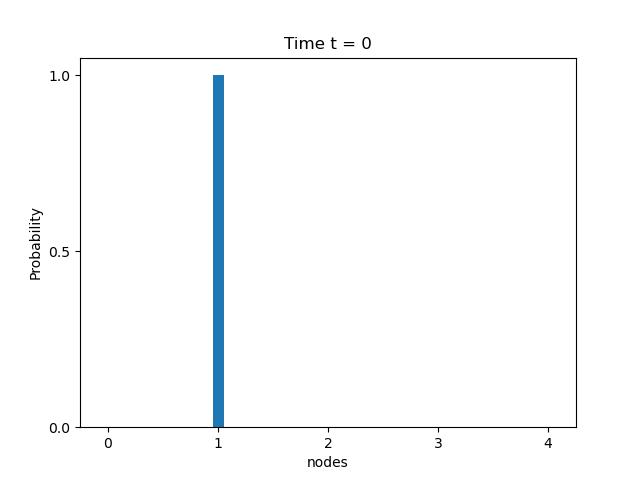}
				\caption{Probability of getting the walker at different vertices at $t = 0$.}
			\end{subfigure}
			\hspace{.5cm}
			\begin{subfigure}[t]{.48\textwidth}
				\centering 
				\includegraphics[height = 4cm, width = 6cm]{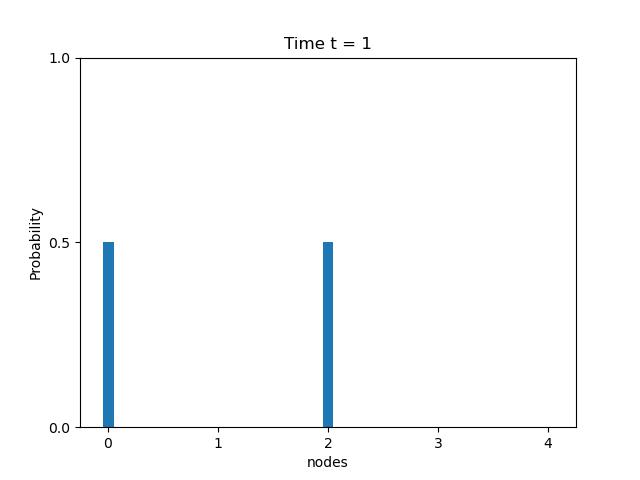}
				\caption{Probability of getting the walker at different vertices at $t = 1$.}
			\end{subfigure}\\
			\begin{subfigure}[t]{.48\textwidth}
				\centering 
				\includegraphics[height = 4cm, width = 6cm]{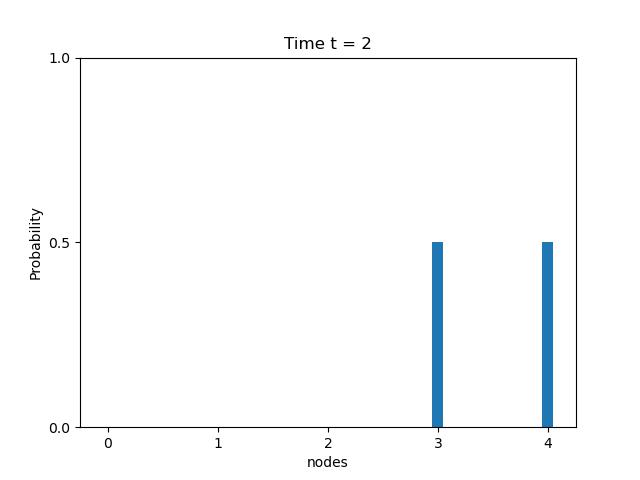}
				\caption{Probability of getting the walker at different vertices at $t = 2$.}
			\end{subfigure}
			\hspace{.5cm}
			\begin{subfigure}[t]{.48\textwidth}
				\centering 
				\includegraphics[height = 4cm, width = 6cm]{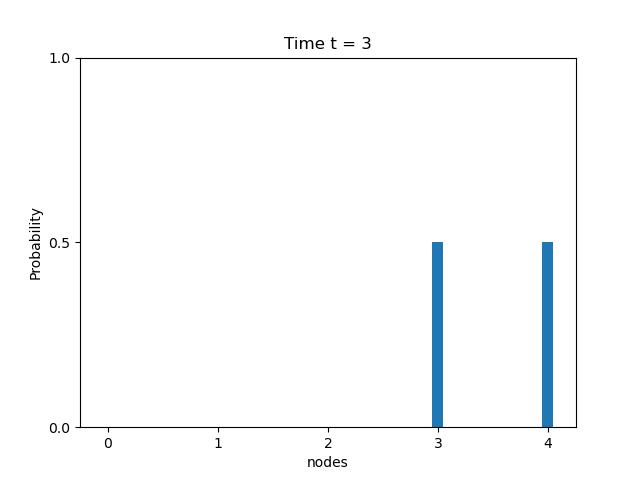}
				\caption{Probability of getting the walker at different vertices at $t = 3$.}
			\end{subfigure}\\
			\begin{subfigure}[t]{.48\textwidth}
				\centering 
				\includegraphics[height = 4cm, width = 6cm]{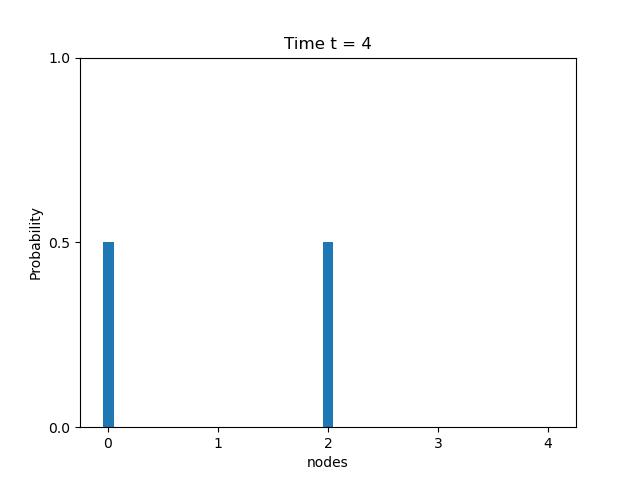}
				\caption{Probability of getting the walker at different vertices at $t = 4$.}
			\end{subfigure}
			\hspace{.5cm}
			\begin{subfigure}[t]{.48\textwidth}
				\centering 
				\includegraphics[height = 4cm, width = 6cm]{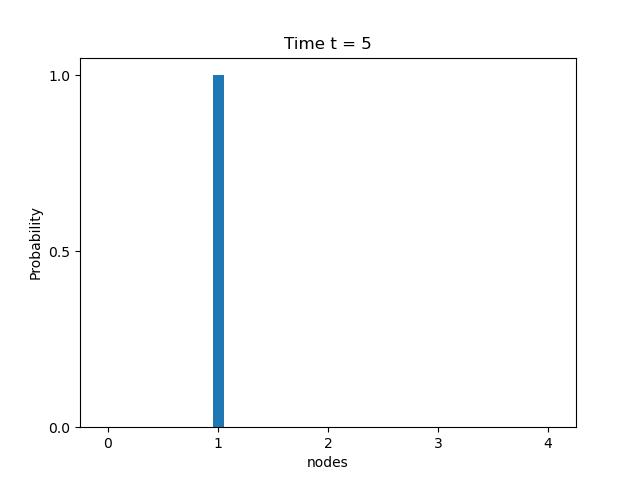}
				\caption{Probability of getting the walker at different vertices at $t = 5$.} 
			\end{subfigure}
			\caption{In this figure, we represent the probability of getting the walker at vertex $k, 0 \leq k \leq 4$ in a cycle graph $C_5$ with $5$ vertices at different time instances $t = 0, 1, \dots, 5$ with bar diagrams. Suppose the walker started at vertex $1$ at $t = 0$. At $t = 5$, it returns to the vertex $1$. Therefore, the probability of getting the walker at time $t = 0$ and $t = 5$ at vertex $1$ is $1$. Hence, the vertex $1$ is periodic with period $t = 5$.}
			\label{cycle_probablity_5}
		\end{figure}

	\section{State transfer on a path graph}
	\label{PST_Path}

		Recall that, a path graph $P_n$ with $n$ vertices $0, 1, \dots (n - 1)$ has edges $(j, (j + 1))$ for $j = 0, 1, \dots (n - 2)$. Now, the out-degree of $0$ and $(n - 1)$ is $1$. Also, the out-degree of each of the remaining vertices is $2$. Applying equation (\ref{probability_distn}) we have 
		\begin{equation}
			p_{j, k} = \begin{cases}
				1 & ~\text{if}~ (j, k) = (0, 1), ~\text{or}~ ((n - 1), (n - 2)); \\
				\frac{1}{2} & ~\text{if}~ (j, k) \in E(G), (j, k) \neq (0, 1), ~\text{or}~ (j, k) \neq ((n - 1), (n - 2));\\
				0 & ~\text{if}~ (j, k) \notin E(G).
			\end{cases}
		\end{equation}
		For every vertex $j$ we generate $p_{j, k}$, such that $\sum_{k \in V(G)} p_{j,k} = 1$. Consider Figure \ref{path5} for an example.
		\begin{figure} 
			\begin{subfigure}[t]{.49\textwidth}
				\centering 
				\begin{tikzpicture}[scale = 1]
					\draw [fill] (0, 0) circle [radius=0.05 cm];
					\node [below] at (0, 0) {0};
					\draw [fill] (1, 0) circle [radius=0.05 cm];
					\node [below] at (1, 0) {1};
					\draw [fill] (2, 0) circle [radius=0.05 cm];
					\node [below] at (2, 0) {2};
					\draw [fill] (3, 0) circle [radius=0.05 cm];
					\node [below] at (3, 0) {3};
					\draw [fill] (4, 0) circle [radius=0.05 cm];
					\node [below] at (4, 0) {4};
					\draw [fill] (5, 0) circle [radius=0.05 cm];
					\node [below] at (5, 0) {5};
					\draw (0, 0) -- (5, 0);
				\end{tikzpicture}
				\caption{}
			\end{subfigure}
			\begin{subfigure}[t]{.49\textwidth}
				\centering 
				\begin{tikzpicture}[scale = 1]
					\draw [fill] (0, 0) circle [radius=0.05 cm];
					\node [below] at (0, 0) {0};
					\draw [fill] (1, 0) circle [radius=0.05 cm];
					\node [below] at (1, 0) {1};
					\draw [fill] (2, 0) circle [radius=0.05 cm];
					\node [below] at (2, 0) {2};
					\draw [fill] (3, 0) circle [radius=0.05 cm];
					\node [below] at (3, 0) {3};
					\draw [fill] (4, 0) circle [radius=0.05 cm];
					\node [below] at (4, 0) {4};
					\draw [fill] (5, 0) circle [radius=0.05 cm];
					\node [below] at (5, 0) {5};
					\draw [out=330,in= 210, ->] (0, 0) to  (.9, 0);
					\draw [out = 150, in = 30, ->](1, 0) to (.1, 0);
					\draw [out=330,in= 210, ->] (1, 0) to  (1.9, 0);
					\draw [out = 150, in = 30, ->](2, 0) to (1.1, 0);
					\draw [out=330,in= 210, ->] (2, 0) to  (2.9, 0);
					\draw [out = 150, in = 30, ->](3, 0) to (2.1, 0);
					\draw [out=330,in= 210, ->] (3, 0) to  (3.9, 0);
					\draw [out = 150, in = 30, ->](4, 0) to (3.1, 0);
					\draw [out=330,in= 210, ->] (4, 0) to  (4.9, 0);
					\draw [out = 150, in = 30, ->](5, 0) to (4.1, 0);
				\end{tikzpicture} 
				\caption{}
				\label{path5}
			\end{subfigure}
			\caption{We have a path graph with $6$ vertices in figure (a). For every edge, we generate two opposite orientations and draw the new graph in figure (b). Except for vertices $0$ and $5$, all other vertices have two outgoing edges. Hence, their out-degree is $2$. The out-degree of $1$ and $5$ is $1$. Therefore $p_{0, 1} = p_{5, 4} = 1$ and for all other edges $p_{j, k} = \frac{1}{2}$.}
		\end{figure}
		
		Applying equation (\ref{psi_j}) we have
		\begin{equation}\label{psi_j_path}
			\ket{\psi_j} =
			\begin{cases}
				\ket{01} & ~\text{if}~ j = 0; \\
				\frac{1}{\sqrt{2}}(\ket{j(j+1)} + \ket{j(j - 1)}) & ~\text{if}~ j = 1, 2, \dots (n - 2); \\
				\ket{(n - 1)(n - 2)} & ~\text{if}~ j = (n - 1).
			\end{cases}
		\end{equation}
		The sum of projection operators on $\mathcal{L}(\Psi)$ for path graph $P_n$ is given by
		\begin{equation}
			\begin{split}
				\Pi = & \ket{01}\bra{01} + \frac{1}{2}\sum_{j = 1}^{(n - 2)} \left( \ket{j(j+1)} + \ket{j(j - 1)} \right) \left(\bra{j(j+1)} + \bra{j(j - 1)} \right) \\
				& + \ket{(n - 1)(n - 2)} \bra{(n - 1)(n - 2)}.
			\end{split}
		\end{equation}
		In the theorem below, we demonstrate PST between the two extreme vertices of any path graph $P_n$.
		
		\begin{theorem}
			There is a PST between the vertices $0$ and $(n - 1)$ in a path graph $P_n$ with $n$ vertices at time $t = (n - 1)$.
		\end{theorem}
		\begin{proof}
			Note that, $[S(2 \Pi - I)] \ket{\psi_0} = \ket{10}$. Applying Lemma \ref{App_Lemma_2} repeatedly on $\ket{10}$ we can prove that $[S(2 \Pi - I)]^t \ket{\psi_0} = \ket{t(t - 1)}$. Putting $t = n - 1$ we have $[S(2 \Pi - I)]^{(n - 1)} \ket{\psi_0} = \ket{(n - 1)(n - 2)} = \ket{\psi_{n - 1}}$. Also, $[S(2 \Pi - I)] \ket{\psi_{n - 1}} = \ket{(n - 2)(n - 1)}$. Applying Lemma \ref{App_Lemma_3} on $\ket{(n - 2)(n - 1)}$ repeatedly, we have $[S(2 \Pi - I)]^t \ket{\psi_{n- 1}} = \ket{(n - t - 1)(n - t)}$. When $t = (n - 1)$ observe that $[S(2 \Pi - I)]^{(n - 1)} \ket{\psi_{n- 1}} = \ket{01} = \ket{\psi_0}$. Therefore, we can conclude that there is a PST between $0$ and $(n -1)$ in $P_n$ at time $t = (n - 1)$.
		\end{proof}
		
		Now, we prove that there is PST between other pairs of vertices in $P_n$.
		
		\begin{theorem}\label{theorem_propagation_on_path}
			Let $1 \leq j \leq n - j - 1$ represents a vertex of a path graph $P_n$, then $[S(2 \Pi - I)]^t \ket{\psi_j} = $
			$$ 
			\begin{cases} 
				& \frac{1}{\sqrt{2}} \left( \ket{(j + t)(j + t - 1)} + \ket{(j - t)(j - (t - 1))} \right) ~\text{when}~ 0 \leq t \leq j; \\
				& \frac{1}{\sqrt{2}} \left( \ket{(j + t)(j+ t - 1)} + \ket{(t - j)(t - j - 1)} \right)   ~\text{when}~ j + 1 \leq t \leq n - j - 1; \\ 
				& \frac{1}{\sqrt{2}} \left( \ket{(2n - t - j - 2)(2n - j - t - 1)} + \ket{(t - j)(t - j - 1)} \right) ~\text{when}~ n - j \leq t \leq (n - 1). 
			\end{cases} $$
		\end{theorem} 
		\begin{proof}
			This result holds trivially when $t = 0$. Lemma \ref{App_Lemma_1} indicates the proof for $t = 1$. For other values of $t$ with $2 \leq t \leq j$, the proof is similar to Theorem \ref{theorem_propagation_on_cycle}. When $t = j$ we have $[S(2 \Pi - I)]^j \ket{\psi_j} = \frac{1}{\sqrt{2}} \left( \ket{(2j)(2j - 1)} + \ket{01} \right)$. 
			
			Now, consider the range $j + 1 \leq t \leq n - j - 1$, that is $1 \leq t_1 = t - j \leq n - 2j -1$.  Applying Lemma \ref{App_Lemma_2}, we have
			\begin{equation}
				\begin{split}
					[S(2 \Pi - I)]^{j + 1} \ket{\psi_j} & = \frac{1}{\sqrt{2}} \left( S(2 \Pi - I) \ket{(2j)(2j - 1)} + S(2 \Pi - I) \ket{01} \right)\\
					&  = \frac{1}{\sqrt{2}} \left( \ket{(2j + 1)(2j)} + \ket{10} \right).
				\end{split}
			\end{equation}
			Applying Lemma \ref{App_Lemma_2} repeatedly we have
			\begin{equation}
				\begin{split}
					& [S(2 \Pi - I)]^{j + t_1} \ket{\psi_j} = \frac{1}{\sqrt{2}} \left( \ket{(2j + t_1)(2j+ t_1 - 1)} + \ket{t_1(t_1 - 1)} \right) \\
					\text{or}~ & [S(2 \Pi - I)]^t \ket{\psi_j} = \frac{1}{\sqrt{2}} \left( \ket{(j + t)(j+ t - 1)} + \ket{(t - j)(t - j - 1)} \right).
				\end{split}
			\end{equation}
			When $t = (n - j - 1)$ we get 
			\begin{equation}
				\begin{split}
					[S(2 \Pi - I)]^{(n - j - 1)} \ket{\psi_j} = \frac{1}{\sqrt{2}} \left( \ket{(n - 1)(n - 2)} + \ket{(n - 2j - 1)(n - 2j - 2)} \right).
				\end{split}
			\end{equation}
			
			Now, consider the range $n - j \leq t \leq (n - 1)$, that is $0 \leq t_2 = t + j - n \leq j + 1$. When $t = n - j$, equivalently, $t_2 = 0$ we have
			\begin{equation}
				\begin{split}
					[S(2 \Pi - I)]^{n - j} \ket{\psi_j} = & \frac{1}{\sqrt{2}} S(2 \Pi - I)\ket{(n - 1)(n - 2)} + \frac{1}{\sqrt{2}} S(2 \Pi - I)\ket{(n - 2j - 1)(n - 2j - 2)} \\
					= & \frac{1}{\sqrt{2}} \left( \ket{(n - 2)(n - 1)} + \ket{(n - 2j)(n - 2j - 1)} \right). 
				\end{split}
			\end{equation} 
			Applying Lemma \ref{App_Lemma_2} and \ref{App_Lemma_3} repeatedly we have
			\begin{equation}
				\begin{split}
					[S(2 \Pi - I)]^{n - j + 1} \ket{\psi_j} = \frac{1}{\sqrt{2}} \left( \ket{(n - 3)(n - 2)} + \ket{(n - 2j + 1)(n - 2j)} \right). 
				\end{split}
			\end{equation} 
			In general we can write 
			\begin{equation}
				\begin{split}
					& [S(2 \Pi - I)]^{n - j + t_2} \ket{\psi_j} \\
					& = \frac{1}{\sqrt{2}} \ket{(n - (t_2 + 2))(n - (t_2 + 1))} + \frac{1}{\sqrt{2}} \ket{(n - 2j + t_2)(n - 2j + t_2 - 1)} \\
					\text{or}~ & [S(2 \Pi - I)]^t \ket{\psi_j} = \frac{1}{\sqrt{2}} \ket{(2n - t - j - 2)(2n - j - t - 1)} + \frac{1}{\sqrt{2}} \ket{(t - j)(t - j - 1)}.
				\end{split}
			\end{equation} 
			Combining all we observe the result.
		\end{proof}
		
		Now, we can calculate the probability $P_t(k)$ of getting the walker at a vertex $k$ on a path graph. Recall that, in general the walker begins its journey from any vertex $j$. Theorem \ref{theorem_propagation_on_path} provides the expression of $U^t \ket{\psi_j}$ and equation (\ref{psi_j_path}) suggests the expression of $\ket{\psi_k}$. Then $P_t(k)$ is generated by the inner product of $\ket{\psi_k}$ and $U^t \ket{\psi_j}$. When the walker starts its journey from $j = 0$, probability of getting the walker at different vertices of a path graph $P_6$ at different time steps are indicated by the green bars in Figure \ref{line_6}.
		
		We mention the following corollaries to discuss PST and periodic vertices in a path graph.
		
		\begin{corollary}
			There is PST between the vertices $j$ and $n - j - 1$ in line graph $P_n$ for $1 \leq j < \frac{n - 1}{2}$ at time $t = n - 1$.
		\end{corollary}
		
		\begin{proof}
			Let the initial state be $\ket{\psi_j}$. Applying Theorem \ref{theorem_propagation_on_path} we have the state at time $t = (n - 1)$, which is 
			\begin{equation}
				[S(2 \Pi - I)]^{n - 1} \ket{\psi_j} = \frac{1}{\sqrt{2}} \left( \ket{(n - j - 1)(n - j)} + \ket{(n - j - 1)(n - j - 2)} \right) = \ket{\psi_{n - j - 1}}.
			\end{equation}
			Therefore there is PST between the vertices $j$ and $(n - j - 1)$.
		\end{proof}
		
		\begin{corollary}
			In a path graph $P_{2m + 1}$ the vertex $m$ is periodic at time $t = (n - 1)$.
		\end{corollary}
		
		\begin{proof}
			Applying $n = 2m + 1$, $j = m$, and $t = n - 1 = 2m$ in Theorem \ref{theorem_propagation_on_path} we have 
			\begin{equation}
				[S(2 \Pi - I)]^{2m} \ket{\psi_m} = \frac{1}{\sqrt{2}} \left( \ket{m(m + 1)} + \ket{m(m - 1)} \right) = \ket{\psi_m}.
			\end{equation}
			Therefore, the vertex $m = \frac{n - 1}{2}$ is periodic at time $t = n - 1$, when $n$ is an odd number.
		\end{proof}
		
		As there is PST s between $j$ and $n - j - 1$ all the vertices in $P_n$ are periodic at time $t = 2(n - 1)$.
		
		We can visualize the movement of the quantum walker in terms of probability. Suppose the walker starts at vertex $j$ at time $t = 0$ with the initial state $\ket{\psi_j}$. Probability of getting it at vertex $k$ at time $t$ is $P_t(j, k) = |\braket{\psi_k | [S(2\Pi - I)]^t | \psi_j}|^2$. Theorem \ref{theorem_propagation_on_path} indicates when $0 \leq t \leq j$ we have 
		\begin{equation}
			P_t(j, k) =  \begin{cases} 1 & ~\text{when}~ k = j ~\text{and}~ t = 0; \\ \frac{1}{4} & ~\text{when}~ k = j + t ~\text{or}~ k = j - t ~\text{and}~ t = 0 \leq j < j - 1; \\ \frac{1}{2} & ~\text{when}~ k = 0 ~\text{and}~ t = j. \end{cases} 
		\end{equation} 
		When $j + 1 \leq t \le n - j - 1$ we have 
		\begin{equation}
			P_t(j, k) =  \begin{cases} \frac{1}{4} & ~\text{when}~ k = j + t ~\text{or}~ k = t - j  ~\text{and}~ j + 1 \leq t < n - j - 1; \\ \frac{1}{4} & ~\text{when}~ k = n - 2j - 1 ~\text{and}~ t = n - j - 1; \\ \frac{1}{2} & ~\text{when}~ k = n - 1 ~\text{and}~ t = n - j - 1. \end{cases} 
		\end{equation}
		Also for $n - j \leq t \leq n - 1$ we have
		\begin{equation}
			P_t(j, k) =  \begin{cases} \frac{1}{4} & ~\text{when}~ k = 2n - t - j - 2 ~\text{or}~ k = t - j ~\text{and}~ n - j \leq t < n - 1; \\ 1 & ~\text{when}~ k = n - j - 1 ~\text{and}~ t = n - 1. \end{cases}
		\end{equation}
		Other than these values of $t, j$ and $k$ we have $P_t(j, k) = 0$. For example, we consider a path graph with $6$ vertices with a walker who initiates walking at vertex $j = 1$. There is a PST between vertex $1$ and $4$. We plot $P_t(j, k)$ with bar diagrams for different values of $t$ in the Figures \ref{line_6}. 
		\begin{figure}
			\begin{subfigure}[t]{.48\textwidth}
				\centering
				\includegraphics[height = 4cm, width = 6cm]{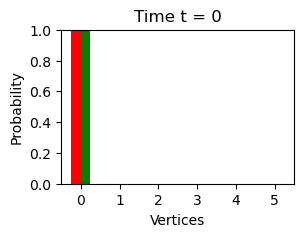}
				\caption{Probability of getting the walker at different vertices at $t = 0$.}
			\end{subfigure}
			\hspace{.5cm}
			\begin{subfigure}[t]{.48\textwidth}
				\centering
				\includegraphics[height = 4cm, width = 6cm]{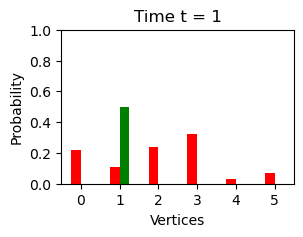}
				\caption{Probability of getting the walker at different vertices at $t = 1$.}
			\end{subfigure}\\
			\begin{subfigure}[t]{.48\textwidth}
				\centering
				\includegraphics[height = 4cm, width = 6cm]{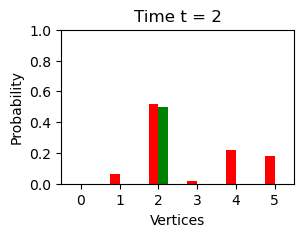}
				\caption{Probability of getting the walker at different vertices at $t = 2$.}
			\end{subfigure}
			\hspace{.5cm}
			\begin{subfigure}[t]{.48\textwidth}
				\centering
				\includegraphics[height = 4cm, width = 6cm]{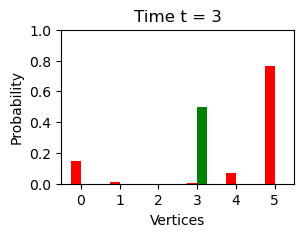}
				\caption{Probability of getting the walker at different vertices at $t = 3$.}
			\end{subfigure}\\
			\begin{subfigure}[t]{.48\textwidth}
				\centering
				\includegraphics[height = 4cm, width = 6cm]{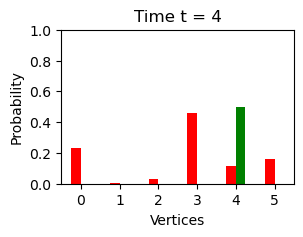}
				\caption{Probability of getting the walker at different vertices at $t = 4$.}
			\end{subfigure}
			\hspace{.5cm}
			\begin{subfigure}[t]{.48\textwidth}
				\centering
				\includegraphics[height = 4cm, width = 6cm]{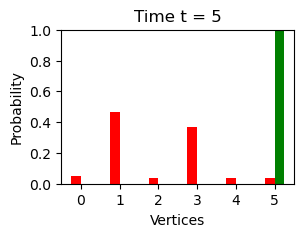}
				\caption{Probability of getting the walker at different vertices at $t = 5$.} 
			\end{subfigure}
			\caption{The red and green bars in the sub-figures indicate the probability of getting the walker at the vertex $k, 0 \leq k \leq 5$ in the path graph with $6$ vertices at different time instances $t = 0, 1, \dots, 5$ in continuous-time and Markovian quantum walks, respectively. Suppose the walker started at vertex $0$ at $t = 0$. At $t = 5$, the walker reaches the vertex $5$ with full probability under Markovian quantum walks. Therefore, there is a PST between vertices $1$ and $4$ at time $t = 5$. But the continuous-time quantum walk does not allow the walker to move to any vertex with full probability.}
			\label{line_6}
		\end{figure}
		
		On the other hand, we can determine the probability of getting the walker at different time steps using the continuous-time quantum walks governed by equation (\ref{continuous_time_walkbased_PST}). Let the walker initiate walking at vertex $0$ at time $t = 0$. Therefore, the probability of getting the walker at $t = 0$ at vertex $0$ is $1$. Probabilities of getting the walker at different vertices in time $t = 1, 2, \dots, 5$ are presented by the red bars in Figure \ref{line_6}. It can be observed that the walker does not move to any vertex with full probability in these time values, following the continuous-time quantum walk.

	\section{State transfer in glued path graph}

		Recall that in Section \ref{PST_Path}, we discussed PST between two extreme vertices $0$ and $(n - 1)$ of a path graph $P_n$ with $n$ vertices. Also, in Section \ref{PST_cycle}, we discuss PST between vertices $0$ and $m$ in a cycle graph $C_{2m}$ with $2m$ vertices. To construct a cycle graph with $2m$ vertices from two path graphs with $m$ vertices, we glue the extreme vertices of the path graphs and update the vertex labeling. We may generalize this process to generate graphs supporting PST between two extreme vertices.
		
		Here, we discuss the gluing process to construct the graphs supporting PST under Markovian quantum walks. Consider $k$ copies of path graph $P_m$ with vertices $0, 1, \dots (m - 1)$ in each. We generate a glued graph $P^{(k)}$ by combining all the $0$ vertex into one vertex. In addition, we combining all the $(m - 1)$ vertices into one vertex. We then relabel the vertices with $0, 1, \dots [k(m - 2) + 1]$. A few glued graphs generated from $P_4$ are depicted in Figure \ref{glued_paths}.
		\begin{figure}
			\begin{subfigure}{.48\textwidth}
				\centering
				\includegraphics[scale = .5]{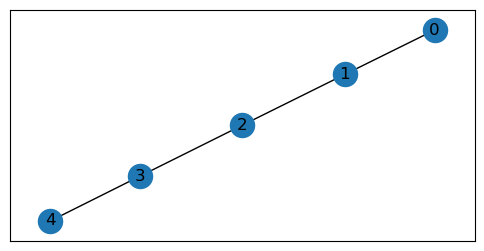}
				\caption{Path graph $P_4$.}
			\end{subfigure}
			\hspace{.5cm}
			\begin{subfigure}{.48\textwidth}
				\centering
				\includegraphics[scale = .5]{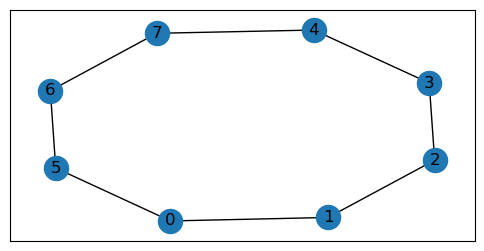}
				\caption{Glued graph $P_4^{(2)}$ from $2$ copies of $P_4$ graphs.}
			\end{subfigure}\\
			\begin{subfigure}{.48\textwidth}
				\centering
				\includegraphics[scale = .5]{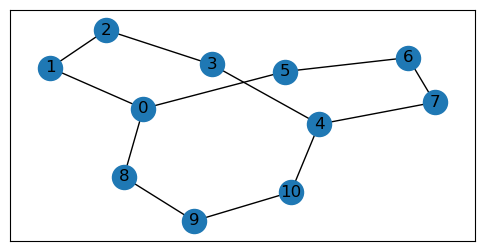}
				\caption{Glued graph $P_4^{(3)}$ from $3$ copies of $P_4$ graphs.}
			\end{subfigure}
			\hspace{.5cm}
			\begin{subfigure}{.48\textwidth}
				\centering
				\includegraphics[scale = .5]{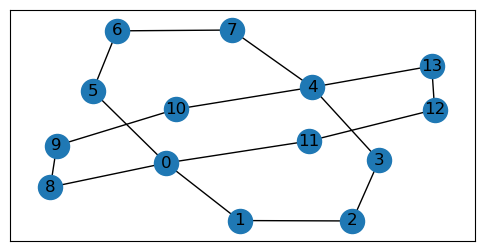}
				\caption{Glued graph $P_4^{(4)}$ from $4$ copies of $P_4$ graphs.}
			\end{subfigure}\\
			\begin{subfigure}{.48\textwidth}
				\centering
				\includegraphics[scale = .5]{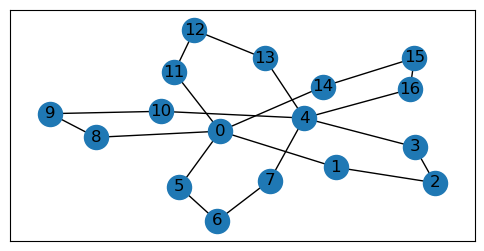}
				\caption{Glued graph $P_4^{(5)}$ from $5$ copies of $P_4$ graphs.}
			\end{subfigure}
			\hspace{.5cm}
			\begin{subfigure}{.48\textwidth}
				\centering
				\includegraphics[scale = .5]{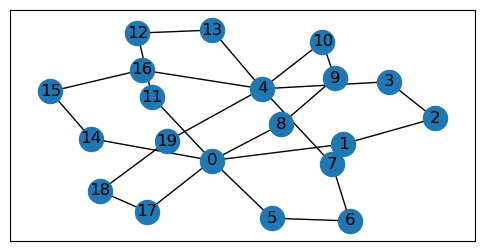}
				\caption{Glued graph $P_4^{(6)}$ from $6$ copies of $P_4$ graphs.}
			\end{subfigure}
			\caption{We consider $k$ copies of path graphs $P_4$ for $k = 2, 3, \dots, 6$. We glue the vertices $0$ and $4$ of all the $k$ copies into two vertices $0$ and $4$ in the resultant graph. Relabeling, we obtain different glued graphs $P_4^{(k)}$ for different values of $k$. Note that, between $0$ and $4$ in a glued graph $P_4^{(k)}$, there are $k$ distinct paths of length $3$.}
			\label{glued_paths}
		\end{figure}
		
		We numerically observe that there is PST between $0$ and $(m - 1)$ in all the graphs. Considering $7$ copies of $P_4$ graphs we construct a glued graph. After staring from the vertex $0$, the walker reaches to vertex $4$ with full probability at time $t = 5$. The probability distribution of getting the walker at different vertices at different time steps under Markovian quantum walks is depicted in Figure \ref{PST_in_glued_graph}.
		\begin{figure}
			\begin{subfigure}{.48\textwidth}
				\includegraphics[scale = .5]{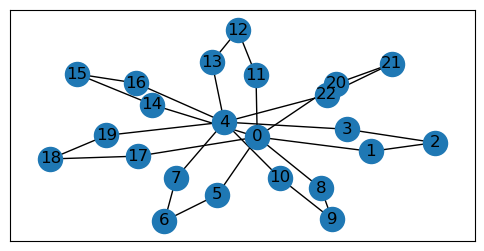}
				\caption{Glued graph $P_4^{(7)}$ from $7$ copies of $P_4$ graph.} 
				\centering 
			\end{subfigure}
			\hspace{.5cm}
			\begin{subfigure}{.48\textwidth}
				\includegraphics[scale = .5]{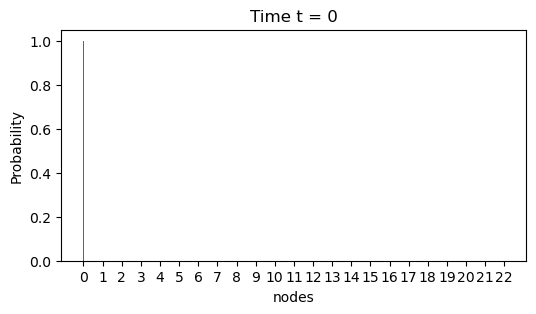}
				\caption{Probability of getting the walker at different vertices at time $t = 0$.} 
				\centering 
			\end{subfigure}\\
			\begin{subfigure}{.48\textwidth}
				\includegraphics[scale = .5]{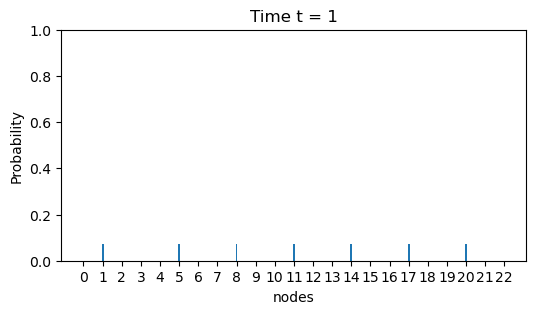}
				\caption{Probability of getting the walker at different vertices at time $t = 1$.} 
				\centering 
			\end{subfigure}
			\hspace{.5cm}
			\begin{subfigure}{.48\textwidth}
				\includegraphics[scale = .5]{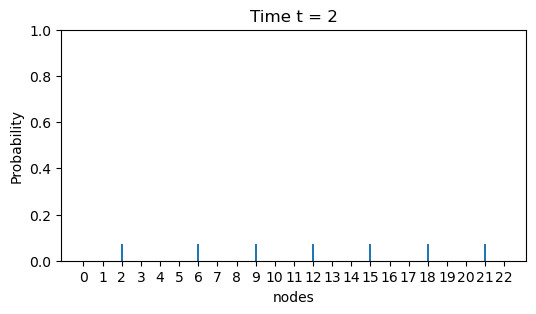}
				\caption{Probability of getting the walker at different vertices at time $t = 2$.} 
				\centering 
			\end{subfigure}\\
			\begin{subfigure}{.48\textwidth}
				\includegraphics[scale = .5]{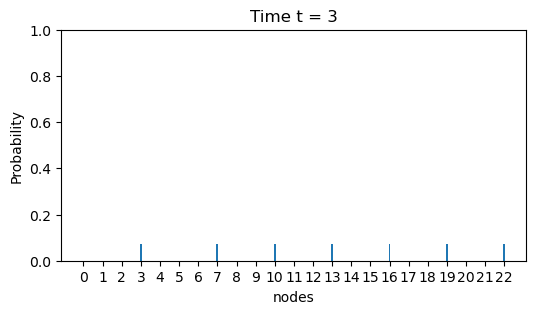}
				\caption{Probability of getting the walker at different vertices at time $t = 3$.} 
				\centering 
			\end{subfigure}
			\hspace{.5cm}
			\begin{subfigure}{.48\textwidth}
				\includegraphics[scale = .5]{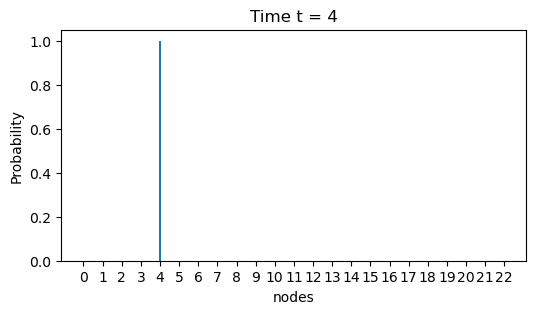}
				\caption{Probability of getting the walker at different vertices at time $t = 4$.} 
				\centering 
			\end{subfigure}
			\caption{In this figure, we represent the probability of getting the walker at vertex $k, 0 \leq k \leq 22$ in a glued path graph $P_4^{(7)}$ with $22$ vertices at different time instances $t = 0, 1, \dots, 4$ with bar diagrams. Suppose the walker started at vertex $0$ at $t = 0$. At $t = 4$, it returns to the vertex $4$.}
			\label{PST_in_glued_graph}
		\end{figure}

	\section{Conclusion}

		In this article, we study PST based on Markovian quantum walk on the path graphs and the cycle graphs with an arbitrary number of vertices. We justify that this PST in cycle and path graphs is more efficient in quantum communication than the original idea of PST. PST based on the continuous-time quantum walk is limited in path graph with at most $3$ vertices and cycle graph with $4$ vertices. On the other hand, the path and cycle graphs with an arbitrary number of vertices exhibit Markovian quantum walk-based PST. Hence, we can overcome this limitation.
		
		The idea of PST based on Markovian quantum walk is new and different from the other proposals of PST. It is well-known that the antipodal vertices of a hypercube graph allow PST based on the continuous-time quantum walk. In our case, only a two-dimensional hypercube graph which is a cycle graph with four vertices allows PST. Hypercube graphs do not allow PST based on Markovian quantum walk in them when their dimension is more than two. Interestingly, all the vertices in the hypercube graph of dimension 4 are periodic at time $t = 12$. The central vertex of the star graphs is periodic at time $t = 2$. The non-central vertices of the star graphs are periodic at time $t = 4$. Tensor product of two path graphs $P_i \otimes P_j$ allows PST when $i + j \leq 7$. Computer programs identify these graphs supporting PST or having periodic vertices in them.
		
		Developing a sufficient condition for PST on general graphs will be an interesting work. An interested reader may attempt it.

	\section*{Funding}
		This work is supported by a project entitled ``Transmission of quantum information using perfect state transfer" (Grant no. CRG/2021/001834) sponsored by the Anusandhan National Research Foundation (ANRF) earlier the Science and Engineering Research Board (SERB).

	\end{onehalfspace}

%	\bibliographystyle{unsrt} 
%	\bibliography{library}

\end{document}